\documentclass[12pt]{article}

\usepackage{amsmath,amssymb,amsthm,amsfonts}
\usepackage{fullpage}
\def\Enn{{\mathbb{N}}}
\def\Reals{{\mathbb{R}}}
\def\Zee{{\mathbb{Z}}}
\def\Cee{{\mathbb{C}}}
\usepackage{fullpage}
\usepackage{float}

\usepackage{pgf}
\usepackage{tikz}
\usetikzlibrary{arrows,automata}

\theoremstyle{plain}
\newtheorem{theorem}{Theorem}
\newtheorem{corollary}[theorem]{Corollary}
\newtheorem{lemma}[theorem]{Lemma}
\newtheorem{proposition}[theorem]{Proposition}

\theoremstyle{definition}
\newtheorem{definition}[theorem]{Definition}

\theoremstyle{remark}

\newcommand{\phiright}{\ \stackrel{\rightarrow}{\varphi^{\omega}} \! }
\newcommand{\hright}{\ \stackrel{\rightarrow}{h^{\omega}} \! }
\newcommand{\muright}{\ \stackrel{\rightarrow}{\mu^{\omega}} \! }

\newcommand{\hleft}{\ \stackrel{\leftarrow}{h^{\omega}} \!\! }

\title{Avoiding Three Consecutive Blocks of the Same Size and Same Sum}

\author{Julien Cassaigne \\
Institut de Math\'ematiques de Luminy\\
Case 907, 163 avenue de Luminy \\
13288 Marseille Cedex 9 \\
France\\
{\tt cassaigne@iml.univ-mrs.fr}\\
\ \\
James D. Currie\\
Department of Mathematics and Statistics \\
University of Winnipeg\\
515 Portage Avenue\\
Winnipeg, Manitoba R3B 2E9 \\
Canada\\
{\tt j.currie@uwinnipeg.ca} \\
\ \\
Luke Schaeffer and Jeffrey Shallit\\
School of Computer Science \\
University of Waterloo\\
Waterloo, ON  N2L 3G1 \\
Canada\\
{\tt l3schaef@student.cs.uwaterloo.ca} \\
{\tt shallit@cs.uwaterloo.ca}
}

\begin{document}

\maketitle

\begin{abstract}
We show that there exists an infinite word over the alphabet
$\lbrace 0, 1, 3, 4 \rbrace$ containing no three consecutive blocks
of the same size and the same sum.  This answers an open problem
of Pirillo and Varricchio from 1994.
\end{abstract}

\section{Introduction}

Avoidability problems in words have received much attention since
the seminal papers of Axel Thue \cite{Thue:1906,Thue:1912,Berstel:1995}.
Generally speaking, the goal is to construct an infinite word over
a finite alphabet with no factor (i.e., a contiguous block of symbols)
having some property, or to show that no such word exists.

Thue constructed an infinite word over a $2$-letter alphabet
containing no factor that is an {\it overlap} (i.e., a finite
word of the form
$axaxa$, where $a$ is a single letter and $x$ is a possibly empty
word), and he also constructed an infinite word
over a $3$-letter alphabet containing no factor that is a {\it square} (i.e.,
a finite nonempty word of the form $ww$).

Thue used iterated morphisms to construct his words.  Given a finite
alphabet $\Sigma$ and a morphism $h:\Sigma^* \rightarrow
\Sigma^*$ satisfying $h(a) = ax$ for some $a \in \Sigma$ and
$x \in \Sigma^*$, we can iterate $h$ to obtain
$$\hright(a) := a \, x \, h(x) \, h^2(x) \, \cdots, $$
which is infinite iff $h^i (x) \not= \epsilon$ for all $i$.  It is
easy to see that $\hright(a)$ is actually a fixed point of $h$; that is,
$h(\hright(a)) = \hright(a)$.  A sufficient (but not necessary)
condition for
$h^i(x) \not= \epsilon$ is that $h$ be {\it nonerasing}, that is,
$h(a) \not= \epsilon$ for all $a \in \Sigma$.

Thue used the morphism $\mu(0) = 01$ and $\mu(1) = 10$.  The fixed point
$\muright(0)$, known as the Thue-Morse word ${\bf t} = 01101001\cdots$,
is overlap-free.  Such
a morphism is called $2$-uniform since each letter is mapped to an image
of size $2$.
It can also be shown that the
fixed point of the (nonuniform) morphism given by $2 \rightarrow 210$,
$1 \rightarrow 20$, $0 \rightarrow 1$, is squarefree.

Erd\H{o}s \cite{Erdos:1961}
introduced the notion of abelian avoidability.  An {\it abelian
$k$-th power} for $k \geq 2$
is a finite nonempty word of the form $ x_1 x_2 \cdots x_k$
where $|x_1| = \cdots = |x_k|$ and each $x_i$ is a permutation of $x_1$.
Dekking \cite{Dekking:1979b}
constructed an infinite word over $\lbrace 0, 1 \rbrace$
containing no abelian $4$-th powers, and an infinite word over a
$3$-letter alphabet containing no abelian cubes.  
The former is given by the fixed point of the morphism 
$a \rightarrow abb$, $b \rightarrow aaab$, and the latter is
given by the fixed point of the morphism
$a \rightarrow aabc$, $b \rightarrow bbc$, and $c \rightarrow acc$.
Ker\"{a}nen \cite{Keranen:1992} constructed an infinite word over a
$4$-letter alphabet containing no abelian squares, using an
$85$-uniform morphism.  In all three cases the
alphabet size is optimal.

In what follows we assume our finite alphabet $\Sigma$ is a subset of
$\Enn$.
An {\it additive $k$-th power} for $k \geq 2$ is a finite nonempty word of the
form $x_1 x_2 \cdots x_k$ where $|x_1| = \cdots = |x_k|$ and 
$\sum x_1 = \sum x_2 = \cdots = \sum x_k$, where by $\sum x_i$ we mean
the sum of the elements appearing in the word $x_i$.
Since two words of the same length over $\lbrace 0, 1 \rbrace$ have
the same sum if and only if they are permutations of each other,
Dekking's result mentioned above shows that it is possible to
avoid additive $4$th-powers.  

In a 1994 paper, Pirillo and Varricchio \cite{Pirillo&Varricchio:1994}
raised the following question:  do there exist infinite words
avoiding additive squares or additive cubes?  They raised the question in the
context of semigroup theory, as follows:

Let $S$ be a semigroup, let $k \geq 1$ be an integer, and
$\varphi:\Sigma^+ \rightarrow S$ be a morphism.  We say that a 
nonempty word $w$
is a {\it uniform $k$-power, mod $\varphi$} if it can be written in the
form $w = w_1 \cdots w_k$ with $\varphi(w_1) = \cdots = \varphi(w_k)$
and $|w_1| = \cdots = |w_k|$.
If there exists an integer $R(\varphi, k)$ such that each word
$w \in \Sigma^+$ with length $\geq R(\varphi, k)$ contains a
factor that is a uniform $k$-power, mod $\varphi$, 
then we say that $\varphi$ is
{\it uniformly $k$-repetitive}.  If for every 
finite alphabet $\Sigma$, every morphism $\varphi:\Sigma^+ \rightarrow S$
is uniformly $k$-repetitive, then we say that
that $S$ is {\it uniformly $k$-repetitive}.  Pirillo and Varricchio proved

\begin{proposition}
Let $k \geq 1$ be an integer.  Then the following are equivalent:

\begin{itemize}
\item[(a)] $\Enn^+$ is not uniformly $k$-repetitive;

\item[(b)] every finitely generated and uniformly $k$-repetitive
semigroup is finite.
\end{itemize}

\end{proposition}

In this context, $\Enn^+$ is not uniformly $k$-repetitive if and only
there exists an infinite word over a finite subset of $\Enn$ that avoids
additive $k$-th powers.   Pirillo and Varricchio observed that, by
Dekking's result mentioned above, 
the semigroup $\Enn^+$ is not uniformly $4$-repetitive.  They remarked,
``We do not know whether $\Enn^+$ is uniformly $2$-repetitive or
uniformly $3$-repetitive.  This seems to be a difficult problem in
combinatorial number theory.''

This theme was taken up again by Halbeisen and
Hungerb\"uhler in 2000 \cite{Halbeisen&Hungerbuhler:2000}, apparently
not knowing of the paper of Pirillo and Varricchio.  They asked 
(in our terminology) if it
is possible to avoid additive squares.

Five other recent papers mentioning the problem of avoiding
additive powers are
\cite{Grytczuk:2008,Richomme&Saari&Zamboni:2009,Cassaigne&Richomme&Saari&Zamboni:2010,Freedman:2010,Au&Robertson&Shallit:2011}.



In this paper we show that there exists an infinite word over the
alphabet $\lbrace 0, 1, 3, 4 \rbrace$ that avoids additive cubes.  This
answers one of the open questions of Pirillo and Varricchio.  As
a consequence we get that $\Enn^+$ is not uniformly $3$-repetitive.

\newcommand{\abs}[1]{\left\lvert#1\right\rvert}
\newcommand{\norm}[1]{\left\lVert#1\right\rVert}
\newcommand{\parent}[1]{\mathrm{par}(#1)} 

\section{Notation and definitions}

We consider the alphabet
$\Sigma = \{ 0, 1, 3, 4 \}$.   Define the morphism 
$\varphi : \Sigma^{*} \rightarrow \Sigma^{*}$ by
\begin{align*}
\varphi(0) &= 03 \\
\varphi(1) &= 43 \\
\varphi(3) &= 1 \\
\varphi(4) &= 01 
\end{align*}
Define the infinite word 
$${\bf w} = \phiright(0) = 031430110343430310110110314303434303434303143011031011011031011011 \cdots .$$
We will show that $\bf w$ contains no additive cubes.

By a {\it block} we will mean a finite factor of $\bf w$.  
The {\it sum} of a block is the sum of its symbols (interpreting
the symbols $0,1,3,4$ as integers). We define a {\it double block} to be a
pair of consecutive blocks, and a {\it triple block} to be a triple of
consecutive blocks. 

\subsection{Matrices and eigenvalues}
Let $\psi : \Sigma^{*} \rightarrow \mathbb \Zee^4$ be the Parikh vector
map, which sends a word $x \in \Sigma^{*}$ to a vector $(|x|_{0},
|x|_{1}, |x|_{3}, |x|_{4})^{T} \in \mathbb \Zee^4$, where $|x|_{a}$ is
the number of occurrences of $a$ in $x$.  We let $M$ denote the
incidence matrix of $\varphi$, given by

\begin{align*}
M &= \begin{pmatrix}
1 & 0 & 0 & 1 \\
0 & 0 & 1 & 1 \\
1 & 1 & 0 & 0 \\
0 & 1 & 0 & 0
\end{pmatrix}.
\end{align*}
Note that $\psi(\varphi(x)) = M \psi(x)$. 

The eigenvalues of $M$ are
the roots of its characteristic polynomial $X^4 - X^3 - 2X^2 + 2X - 1$ and
are, to limited precision,
as follows:\footnote{In this paper, without further comment,
we will frequently make use of floating
point approximations to certain algebraic numbers.  We leave it to the
reader to verify that the approximations we use are accurate enough to
verify our claims.}
\begin{align*}
\lambda_1 &\doteq 1.690284494616614 \\
\lambda_2 &\doteq -1.505068413621472 \\
\lambda_3 &\doteq 0.407391959502429 + 0.476565325929643i \\
\lambda_4 &\doteq 0.407391959502429 - 0.476565325929643i .
\end{align*}
Let $\Lambda$ denote the diagonal matrix of eigenvalues.

The eigenvectors of M are given by columns of the following matrix,
\begin{align*}
Q &=
\begin{pmatrix} 
0.47239594 & 0.17807189  & 0.62696309 & 0.62696309 \\
0.55118080 & 0.67138434  & -0.29375620-0.05534050i & -0.29375620+0.05534050i \\
0.60556477 & -0.56439708 & 0.27824282-0.46132816i & 0.27824282+0.46132816i \\
0.32608759 & -0.44608227 & -0.37154337+0.29878887i & -0.37154337-0.29878887i 
\end{pmatrix}.
\end{align*}
The eigenvectors are normalized to have Euclidean norm 1. Together these
matrices are an eigenvalue decomposition of $M$ since $M = Q \Lambda Q^{-1}$.

We make extensive use of this decomposition. In particular, let $\tau : \Cee^4
\rightarrow \Cee^4$ be the linear map corresponding to left-multiplication by
$Q^{-1}$. Also define linear maps $\tau_j : \Cee^4 \rightarrow \Cee$ for
$j=1,2,3,4$ such that $\tau(x) = (\tau_1(x), \tau_2(x), \tau_3(x), \tau_4(x))$.
Then since $Q^{-1} M = \Lambda Q^{-1}$, we have $\tau_j (M \mathbf{x}) =
\lambda_j \tau_j(\mathbf{x})$ for all vectors $\mathbf{x}$.

The matrix for $\tau$ is just $Q^{-1}$, and we have
\begin{align*}
Q^{-1} &\doteq \begin{pmatrix}
   0.5124 					&   0.5979 						&   0.3537 						&   0.6569 \\
   0.1806 					&   0.6809 						&  -0.4524 						&  -0.5724 \\
   0.5788 - 0.5749i &  -0.3219 + 0.2183i 	&  -0.0690 + 0.6165i 	&  -0.1662 - 0.6810i \\
   0.5788 + 0.5749i &  -0.3219 - 0.2183i 	&  -0.0690 - 0.6165i 	&  -0.1662 + 0.6810i \\
\end{pmatrix}.
\end{align*}
The rows of this matrix give us the maps $\tau_1, \ldots, \tau_4$. 
Thus, for example, $\tau_1 (a,b,c,d) \doteq .5124 + .5979b + .3537c + .6569d$.

\subsection{Indexing and parents}

Let ${\bf w}[i]$ denote the $i$th symbol of ${\bf w}$,
with ${\bf w}[0] = 0$ being the
first symbol of $\bf w$. We let ${\bf w}[p, q)$ denote the symbols from $p$ to
$q$ excluding the symbol at $q$, as long as $p \leq q$. We interpret
${\bf w}[p, p)$ to be the empty word.

Define the function $\eta$ that maps a position $p$ to
$|\varphi({\bf w}[0, p))|$.
Since ${\bf w} = \varphi({\bf w})$, the morphism
$\varphi$ maps any prefix ${\bf w}[0, p)$ of $\bf w$ to
some other prefix of $\bf w$. Therefore $\varphi({\bf w}[0, p))$ is the unique
prefix of length $\eta(p)$, that is, ${\bf w}[0, \eta(p))$. For example, 
${\bf w}[0,3) = 031$, and ${\bf w}[0,3) = 031$ maps to $\varphi(031) =
03143 = {\bf w}[0,5)$, and hence $\eta(3) = 5$.

Note that $\varphi$ is nonerasing,
so it follows that $|\varphi(x)| \geq
|x|$. Since $\varphi(0) = 03$, it follows that $|\varphi(x)| > |x|$ for
any nonempty prefix $x$ of $\bf w$. Hence $\eta(p) \geq p$ for
all $p$, and the inequality is strict for $p > 0$. The function
$\eta$ is also clearly a non-decreasing function, so $a \leq b$ implies
$\eta(a) \leq \eta(b)$. We also have
\begin{align*}
{\bf w}[0, \eta(p)) \ {\bf w}[\eta(p), \eta(p+1)) &= {\bf w}[0, \eta(p+1)) \\
&= \varphi({\bf w}[0, p+1)) \\
&= \varphi({\bf w}[0, p)) \varphi({\bf w}[p]) \\
&= {\bf w}[0, \eta(p)) \varphi({\bf w}[p]),
\end{align*}
so ${\bf w}[\eta(p), \eta(p+1)) = \varphi({\bf w}[p])$.
Thus, the image of ${\bf w}[p]$ is ${\bf w}[\eta(p), \eta(p+1))$.

By definition, some position $p$ maps to $\varphi({\bf w}[p])$ starting
at $\eta(p)$, so we can think of each symbol in ${\bf w}[\eta(p), \eta(p+1))$
as arising from ${\bf w}[p]$. We define a function to associate the positions
in $[\eta(p), \eta(p+1))$ with $p$, given below.
\begin{definition}
For a position $p$ in $\bf w$, we let $\parent{p}$ denote the \emph{parent}
of $p$, which we define to be the unique position $t$ such that
$\eta(t) \leq p < \eta(t+1)$. Also, a \emph{child} of
a position $p$ is any position $q$ such that $\parent{q} = p$.
\end{definition}
Parents have two elementary properties, which we present without proof. 
\begin{enumerate}
\item The inequality $\parent{p} \leq p$ holds for all $p$ with strict inequality unless $p = 0$. 
\item If $a \leq b$ then $\parent{a} \leq \parent{b}$.  In other words, $\parent{x}$ is a non-decreasing function. 
\end{enumerate}

The following table illustrates these concepts for the first few
positions:
\begin{table}[H]
\begin{center}
\begin{tabular}{|c|c|c|c|c|c|c|c|c|c|c|c|c|c|c|c|c|}
\hline
$p$          & 0 & 1 & 2 & 3 & 4 & 5 & 6 & 7 & 8 & 9 & 10 & 11 & 12 & 13 & 14 & 15 \\
\hline
${\bf w}[p]$ & 0 & 3 & 1 & 4 & 3 & 0 & 1 & 1 & 0 & 3 &  4 &  3 &  4 &  3 &  0 &  3 \\
\hline
$\eta(p)$    & 0 & 2 & 3 & 5 & 7 & 8 &10 &12 &14 &16 & 17 & 19 & 20 & 22 & 23 & 25 \\
\hline
$\parent{p}$ & 0 & 0 & 1 & 2 & 2 & 3 & 3 & 4 & 5 & 5 &  6 &  6 &  7 &  7 &  8 &  8 \\
\hline
\end{tabular}
\end{center}
\end{table}

We now form an infinite graph $\mathcal{T}$ with positions as vertices and
edges from each vertex to its children (in the position sense). It follows from
these properties that there is a path in $\mathcal{T}$ from $0$ to any
vertex. Also $\mathcal{T}$ is acyclic with the exception of the loop
at 0. In other words, with the exception of a single loop, $\mathcal{T}$ is an
infinite tree with 0 at the root. Part of this tree is shown in
Figure~\ref{fig:tree}, where we see that $\mathbf{w}$ is obtained by a
level-order traversal of $\mathcal{T}$. Indeed, the levels are equal to $a =
0$, $x = 3$, $\varphi(x) = 1$, $\varphi^2(x) = 43$, etc. so that 
$$ w = a \, x \, \varphi(x) \, \varphi(x)^2 \, \varphi(x)^3 \, \ldots.$$ 

Since $\mathcal{T}$ is a tree, we are often interested in the path from
0 (the root) to an arbitrary vertex. We define the ancestral sequence
of a position $p$ to be the sequence $\{ p_i \}_{i=0}^{\infty}$ such
that $p_0 = p$ and $p_{i+1} = \parent{p_{i}}$ for all $i \geq 0$.

\begin{figure}[htbp]
\centering
\begin{tikzpicture}[>=latex,text height=1.5ex,text depth=0.25ex]
  \matrix[row sep=1cm,column sep=0.5cm] {
  		& &
  		\node (w0) {$w[0]=0$}; & & & & \\
  		& &
  		\node (w1) {$w[1]=3$}; & & & & \\
  		& &
  		\node (w2) {$w[2]=1$}; & & & & \\

		& &
  		\node (w3) {$w[3]=4$}; 
		&
  		\node (w4) {$w[4]=3$}; & & & \\

		&
  		\node (w5) {$w[5]=0$}; & 
  		\node (w6) {$w[6]=1$}; & 
		&
  		\node (w7) {$w[7]=1$}; & 
		& \\

  		\node (w8) {$w[8]=0$}; & 
  		\node (w9) {$w[9]=3$}; & 
  		\node (w10) {$w[10]=4$}; & 
  		\node (w11) {$w[11]=3$}; & 
  		\node (w12) {$w[12]=4$}; & 
  		\node (w13) {$w[13]=3$}; \\
    };
    
	\draw[->,loop above](w0) to node[above]{$\varepsilon$} (w0);
	\draw[->](w0) to node[left]{$0$} (w1);
	\draw[->](w1) to node[left]{$\varepsilon$} (w2);
	\draw[->](w2) to node[left]{$\varepsilon$} (w3);
	\draw[->](w2) to node[right]{$4$} (w4);
	\draw[->](w3) to node[left]{$\varepsilon$} (w5);
	\draw[->](w3) to node[right]{$0$} (w6);
	\draw[->](w4) to node[right]{$\varepsilon$} (w7);
	\draw[->](w5) to node[left]{$\varepsilon$} (w8);
	\draw[->](w5) to node[right]{$0$} (w9);
	\draw[->](w6) to node[left]{$\varepsilon$} (w10);
	\draw[->](w6) to node[right]{$4$} (w11);
	\draw[->](w7) to node[left]{$\varepsilon$} (w12);
	\draw[->](w7) to node[right]{$4$} (w13);
\end{tikzpicture}
\caption{The first 6 levels of $\mathcal{T}$}
\label{fig:tree}
\end{figure}
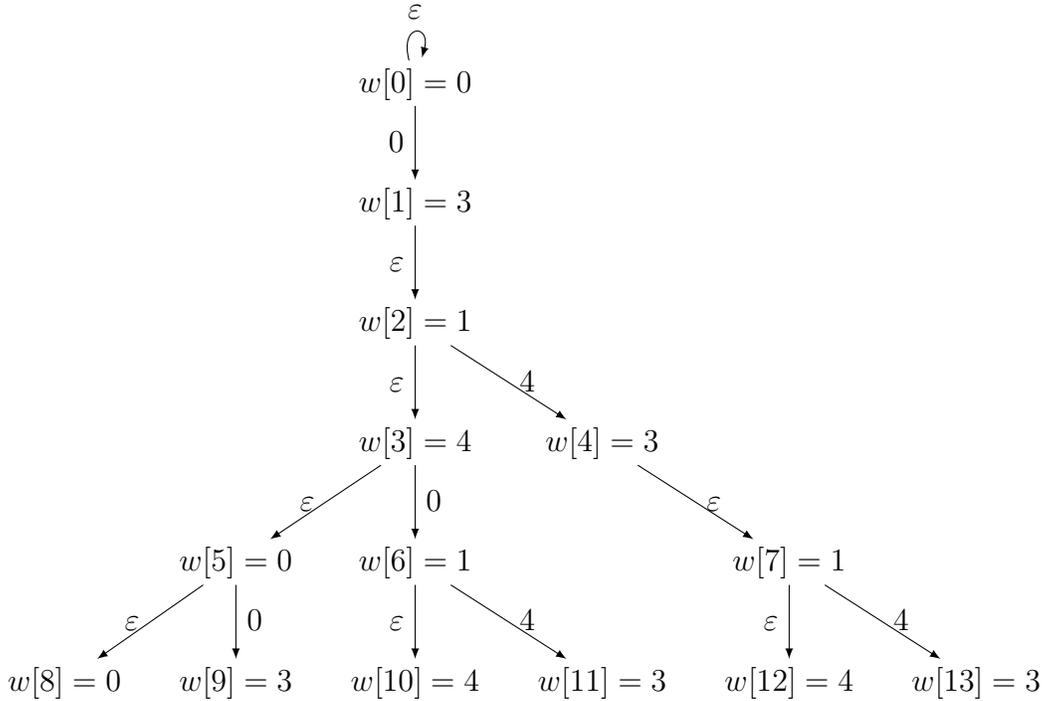

Suppose we are given a vector of positions $\mathbf{p} = (p_1, \ldots,
p_k)$ such that $p_1 \leq p_2 \leq \cdots \leq p_k$. Then $p_i$ and
$p_{i+1}$ delimit the block $b = {\bf w}[p_i, p_{i+1})$ for each $1 \leq i <
k$, so we have $k-1$ consecutive blocks. We extend the definition of
parents to vectors by the equation
\begin{align*}
\parent{p_1, \ldots, p_k} := (\parent{p_1}, \ldots, \parent{p_k}).
\end{align*}
Given consecutive blocks $b_1 \cdots b_{k-1}$ delimited by $\mathbf{p}$,
define $\parent{b_1 \cdots b_{k-1}}$ to be the blocks delimited by
$\parent{\mathbf{p}}$. We also extend the definition of ancestral
sequence to be the sequence of iterated parents for anything that has
parents, e.g., positions, vectors of positions, and consecutive blocks.

\subsection{Parikh vectors of prefixes and blocks}
Define the function $\sigma(p) := \psi({\bf w}[0, p))$. That is,
$\sigma(p)$ is the Parikh vector of the prefix of $\bf w$ up to, but
not including, the position $p$. Note that
\begin{align*}
\sigma(p) \cdot (1,1,1,1) &= |{\bf w}[0, p)|_{0} + |{\bf w}[0, p)|_{1} + |{\bf w}[0, p)|_{3} +
|{\bf w}[0, p)|_{4} \\
&= |{\bf w}[0,p)| = p.
\end{align*}
It follows that $\sigma$ is injective. 

If $(p,q)$ is an edge in $\mathcal{T}$ then ${\bf w}[0, q)$ contains
$\varphi({\bf w}[0, p))$ and perhaps another symbol, so we expect that
$\sigma(q) \approx M \sigma(p)$. The following lemma makes this
precise.
\begin{lemma} \label{sigmarec}
Given a position $p$, there is a bijection between children of $p$ and 
proper prefixes of $\varphi({\bf w}[p])$.  (By a {\it proper} prefix of a word $x$, we mean a 
possibly empty prefix different from $x$.)
Furthermore, if $q$ is a child of $p$ and $a$ is the corresponding
prefix of $\varphi({\bf w}[p])$ then we have
\begin{align*}
\sigma(q) &= M \sigma(p) + \psi(a) .
\end{align*}
\end{lemma}

\begin{proof}
Recall that $\varphi$ maps the symbol ${\bf w}[p]$ to ${\bf w}[\eta(p),
\eta(p+1))$. By definition, $q$ is a child of $p$ if and only if
$\eta(p) \leq q < \eta(p+1)$. Then
\begin{align*}
\sigma(q) &= \psi({\bf w}[0, q)) \\
&= \psi({\bf w}[0,\eta(p)) {\bf w}[\eta(p),q)) \\
&= \psi({\bf w}[0,\eta(p))) + \psi({\bf w}[\eta(p),q)) .
\end{align*}
Let $a = {\bf w}[\eta(p), q)$,
a proper prefix of ${\bf w}[\eta(p), \eta(p+1)) = \varphi({\bf w}[p])$. Then
\begin{align*}
\sigma(q) &= \psi(\varphi({\bf w}[0, p))) + \psi(a) \\
&= M \psi({\bf w}[0, p)) + \psi(a) \\
&= M \sigma(p) + \psi(a) .
\end{align*}
\end{proof}

Thus, every edge $(p, q)$ in $\mathcal{T}$ has a corresponding word
$a$, as proper prefix of $\varphi({\bf w}[p])$. We can think of the $a$
corresponding to an edge as an edge label. This allows us to extend the
previous lemma from a single edge to any walk in $\mathcal{T}$.

\begin{corollary}
\label{sigmarecwalks}
If $p_0 \cdots p_{\ell}$ is a walk in $\mathcal{T}$ with edges $a_{1}, \ldots, a_{\ell}$ then 
\begin{align*}
\sigma(p_\ell) &= \sum_{i=1}^{\ell} M^{\ell - i} \psi(a_{i}) + M^\ell \sigma(p_0).
\end{align*}
\end{corollary}
\begin{proof}
We use induction on $\ell$, the length of the walk. When $\ell = 0$ the result is trivial.
Otherwise, by Lemma~\ref{sigmarec} we have
\begin{align*}
\sigma(p_{\ell}) &= M \sigma(p_{\ell-1}) + \psi(a_{\ell}) .
\end{align*}
Then we apply the induction hypothesis and simplify to complete the proof:
\begin{align*}
\sigma(p_{\ell}) &= M \left( \sum_{i=1}^{\ell - 1} M^{\ell - 1 - i} \psi(a_{i}) + M^{\ell-1} \sigma(p_0) \right) + M^0 \psi(a_{\ell}) \\
&= \sum_{i=1}^{\ell} M^{\ell - i} \psi(a_{i}) + M^{\ell} \sigma(p_0).
\end{align*}
\end{proof}
Now suppose we apply this corollary to an ancestral sequence, $\{ p_i \}_{i=0}^{\infty}$. Let $a_i$ be the label for the edge from $p_i$ to $p_{i+1}$, for each $i$. Then the corollary says that 
\begin{align*}
\sigma(p_0) &= \sum_{i=0}^{k-1} M^i \psi(a_i) + M^{k} \sigma(p_k).
\end{align*}
For $k$ large enough we get $p_k = 0$ so $\sigma(p_k) = 0$ and thus 
\begin{align}
\label{maineqn}
\sigma(p_0) &= \sum_{i=0}^{\infty} M^i \psi(a_i).
\end{align}

\subsection{A graph homomorphism}

Define a directed graph $\mathcal{Q} = (\Sigma, T)$ where vertices are
symbols, and with a set of labelled edges $T$ as shown in
Figure~\ref{fig:dfa}. Let there be an edge from $c \in \Sigma$ to $d
\in \Sigma$ labelled by $\ell \in \Sigma^{*}$ whenever $\ell$ is a
prefix of $\varphi(c)$ up to, but not including, some symbol $d$ in
$\varphi(c)$.

\begin{figure}[htp]
\begin{center}

\begin{tikzpicture}[->,>=stealth',shorten >=1pt,auto,node distance=2.8cm,
	                    semithick]
	  \tikzstyle{every state}=[circle]
	
	  \node[state] 				 (0)                    {0};
	  \node[state]         (4) [right of=0] 			{4};
	  \node[state]				 (1) [below of=4]				{1};
	  \node[state]				 (3) [below of=0] 			{3};
	  
	  \path (0) edge [min distance=15 mm] node {$\varepsilon$} (0)
	  			(0) edge						  				node {$0$} (3)
	  			(4) edge 				              node {$\varepsilon$} (0)
	  			(4) edge [bend left=10]			  node {$0$} (1)
	  			(1) edge [bend left=10]				node {$\varepsilon$} (4)
	  			(1) edge [bend left=10]				node {$4$} (3)
	  			(3) edge [bend left=10]				node {$\varepsilon$} (1);
\end{tikzpicture}
\end{center}
\caption{The directed graph $\mathcal{Q}$}
\label{fig:dfa}
\end{figure}
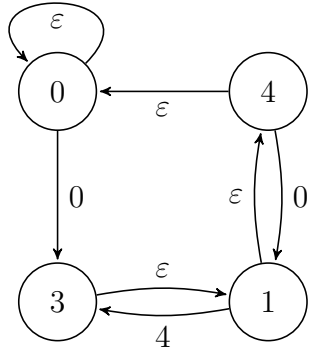

Notice that the map $\zeta$ that sends $x$ to ${\bf w}[x]$ maps vertices in
$\mathcal{T}$ to vertices in $\mathcal{Q}$. Then Lemma~\ref{sigmarec} says that
if an edge $(p, q)$ is labelled by $a$ then the word $a\zeta(q)$ 
is a prefix of $\varphi(\zeta(p))$, so there is an edge $(\zeta(p),
\zeta(q))$ in $\mathcal{Q}$ also labelled by $a$. Therefore $\zeta$ is a graph
homomorphism from $\mathcal{T}$ to $\mathcal{Q}$ that preserves edge labels.

\begin{definition}
We say a labelled digraph homomorphism $f \colon \mathcal{A} \rightarrow \mathcal{B}$ is \emph{child-bijective} if 
$f$ maps children of $a$ to children of $f(a)$ bijectively 
for all vertices $a \in \mathcal{A}$.
\end{definition}

We claim that $\zeta$ is child-bijective. Every child of $\zeta(p)$ in
$\mathcal{Q}$ corresponds to a prefix of $\varphi({\bf w}[p])$, and these
prefixes correspond to children of $p$ according to Lemma~\ref{sigmarec},
therefore children of $\zeta(p)$ correspond to children of $p$. Furthermore, if
$q$ is a child of $p$ then $\zeta(q)$ is a child of $\zeta(p)$, so $\zeta$ is
indeed child-bijective. Child-bijectivity implies a bijection between walks
starting at $p$ and $\zeta(p)$ respectively. 
\begin{proposition}
\label{bijectwalk}
Let $f \colon \mathcal{A} \rightarrow \mathcal{B}$ be a child-bijective labelled digraph homomorphism. We define a function $\hat{f}$ that sends walks in $\mathcal{A}$ to walks in $\mathcal{B}$ so that 
\begin{align*}
v_1 \cdots v_k \mapsto f(v_1) \cdots f(v_k).
\end{align*}
If we fix some $v \in \mathcal{A}$ then $\hat{f}$ is bijection between walks starting at $v$ in $\mathcal{A}$ and walks starting at $f(v)$ in $\mathcal{B}$.
\end{proposition}
\begin{proof}
Fix the vertex $v$ in $\mathcal{A}$. Let $X_\ell$ be the set of walks in $\mathcal{A}$ starting at $v$ of length $\ell$ and likewise let $Y_\ell$ be the set of walks in $\mathcal{B}$ starting at $f(v)$ of length $\ell$. Clearly $\hat{f}$ maps walks of length $\ell$ to walks of length $\ell$, so it suffices to show that $\hat{f}$ restricts to a bijection between $X_\ell$ and $Y_\ell$ for each $\ell$.

Our proof proceeds by induction on $\ell$. There is only one element in both $X_{0}$ and $Y_{0}$, so $\hat{f}$ restricts to a bijection between $X_0$ and $Y_0$. Let $v_0 \cdots v_\ell$ be a walk in $X_\ell$. We decompose this walk into a shorter walk $v_0 \cdots v_{\ell-1}$ in $X_{\ell-1}$ and the final edge $(v_{\ell-1}, v_{\ell})$. By induction, $\hat{f}$ is a bijection between $X_{\ell-1}$ and $Y_{\ell-1}$, so the walk maps to $f(v_0) \cdots f(v_{\ell-1})$. Since $f$ is child-bijective, $f$ maps neighbours of $v_{\ell-1}$ to neighbours of $f(v_{\ell-1})$, so $v_{\ell}$ goes to $f(v_{\ell})$. Now we recompose $f(v_0) \cdots f(v_{\ell-1})$ and the edge $(f(v_{\ell-1}), f(v_{\ell}))$, to give the walk $f(v_0) \cdots f(v_{\ell}) = \hat{f}(v_0 \cdots v_{\ell})$. Since all the maps were bijective, $\hat{f}$ is indeed a bijection. 
\end{proof}
By this proposition, $\zeta$ gives us a bijection between walks in $\mathcal{T}$ starting at some position $p$ and walks in $\mathcal{Q}$ starting at ${\bf w}[p]$. If we are only interested in the edges of a walk in $\mathcal{T}$, then we can use an equivalent walk in $\mathcal{Q}$. The following definition illustrates this idea.

\begin{definition}
Define the set of vectors $\mathcal{D}_{\ell} \subseteq \Zee^4$ by
\begin{align*}
\mathcal{D}_{\ell} &= \left\{ \sum_{i=0}^{\ell-1} M^i \psi(a_i) \ : \ \text{$a_{\ell-1} \cdots a_0$ the 
edges of a walk in $\mathcal{T}$} \right\}.
\end{align*}
\end{definition}
By Proposition~\ref{bijectwalk}, this is equivalent to the following alternate definition of $\mathcal{D}_{\ell}$:
\begin{align*}
\mathcal{D}_{\ell} &= \left \{ \sum_{i=0}^{\ell-1} M^{i} \psi(a_{i}) \ : \ \text{$a_{\ell-1} \cdots a_0$ the edges of a walk in $\mathcal{Q}$} \right \} .
\end{align*}
Later we will need to find all elements of $\mathfrak{D}_{9}$, and this
second definition is important because there are only finitely many
walks of length 9 in $\mathcal{Q}$ compared to infinitely many in
$\mathcal{T}$. Thus, it is straightforward to enumerate the walks in
$\mathcal{Q}$, and thus elements of $\mathfrak{D}_{9}$.

\section{Comparing block sequences} \label{sec:seq}

In this section we are concerned with blocks $b_0$ and $c_0$.
Beginning in Section~\ref{subsec:intersec},
we take $b_0$ and $c_0$ to be blocks
with the same length and sum, and let $\{ b_i \}_{i=0}^{\infty}$ and
$\{ c_i \}_{i=0}^{\infty}$ be the corresponding
ancestral sequences. The
goal of this section is to show that $\psi(b_i)$ and $\psi(c_i)$, the
Parikh vectors of corresponding ancestors for $b_0$ and $c_0$, are
approximately the same. That is, $\psi(b_i) - \psi(c_i)$ is bounded by
a constant that does not depend on $b_0$ and $c_0$.

There are four main steps to this proof:
\begin{enumerate}
\item We bound $\abs{\tau_3(\psi(b_i) - \psi(c_i))}$ and $\abs{\tau_4(\psi(b_i) - \psi(c_i))}$, proving that $\psi(b_i) - \psi(c_i)$ is close to a 2-dimensional subspace. 
\item We show that the sum and length conditions force $\psi(b_0) - \psi(c_0)$ to be in a lattice, $\mathfrak{L}$. We also show that the intersection of this lattice with the 2-dimensional subspace is trivial, and hence $\psi(b_0) - \psi(c_0)$ belongs to a finite set of points. 
\item We bound $\abs{\tau_1(\psi(b_i) - \psi(c_i))}$ and $\abs{\tau_2(\psi(b_i) - \psi(c_i))}$. 
\item We show that since all the eigencoordinates are small, $\psi(b_i) - \psi(c_i)$ is short, and we discuss how to enumerate these short vectors for the next section. 
\end{enumerate}

\subsection{Bounding two coordinates} \label{sec:thirdandfourth}

We start by bounding the third and fourth eigencoordinates. For this
step we do not require that the blocks have the same length and sum,
so we state the theorem for any two blocks.

\begin{theorem}
If $b$ and $c$ are blocks (not necessarily consecutive) then 
\begin{align*}
\abs{\tau_3(\psi(b) - \psi(c))} &\leq C_3 \\
\abs{\tau_4(\psi(b) - \psi(c))} &\leq C_3
\end{align*}
where $C_3 \doteq 2.1758$ is a constant.
\label{thm34}
\end{theorem}
\begin{proof}
First, notice that $\tau_3$ is the complex conjugate of $\tau_4$. That is, for any $x \in \Reals^4$ we have $\tau_3(x) = \overline{\tau_4(x)}$. Therefore we only need to prove one of the inequalities because
\begin{align*}
\abs{\tau_3(\psi(b) - \psi(c))} &= \abs{\tau_4(\psi(b) - \psi(c))}.
\end{align*}
We will prove the first inequality. 
 
From Eq.~(\ref{maineqn}) we see that
$$ \sigma(p_0) = \sum_{i=0}^\infty M^i \psi(a_i).$$
If our block $b$ is delimited by $p_0$ and $q_0$, where 
\begin{align*}
\sigma(q_0) &= \sum_{i=0}^{\infty} M^i \psi(a_i'),
\end{align*}
then we let $\delta_i = \psi(a_i') - \psi(a_i)$ and get 
\begin{align*}
\psi(b) &= \sigma(q_0) - \sigma(p_0) \\
&= \sum_{i=0}^{\infty} M^{i} \psi(a_i') - \sum_{i=0}^{\infty} M^{i} \psi(a_i) \\
&= \sum_{i=0}^{\infty} M^{i} \delta_i .
\end{align*}  
Now apply $\tau_3$ to get
\begin{align*}
\tau_3(\psi(b)) &= \tau_3 \left( \sum_{i=0}^{\infty} M^{i} \delta_i \right) = \sum_{i=0}^{\infty} \lambda_3^i \tau_3(\delta_i).
\end{align*}
We consider the magnitude and separate the first 9 terms from the rest:
\begin{align*}
\abs{\tau_3(\psi(b))} &= \abs{\sum_{i=0}^{\infty} \lambda_3^i \tau_3(\delta_i)} \\
&\leq \abs{\sum_{i=0}^{8} \lambda_3^i \tau_3(\delta_i)} + \abs{\sum_{i=9}^{\infty} \lambda_3^i \tau_3(\delta_i)} .
\end{align*}
We bound the two parts separately, using different techniques. For the finite 
sum we have 
\begin{align*}
\abs{\tau_3 \left( \sum_{i=0}^{8} M^{i} \delta_i \right)} &= \abs{\tau_3 \left( \sum_{i=0}^{8} M^i \psi(a_i') \right) - \tau_3 \left(\sum_{i=0}^{8} M^{i} \psi(a_i) \right)} = \abs{\tau_{3}( \alpha' ) - \tau_{3}(\alpha) },
\end{align*}
where $\alpha = \sum_{i=0}^{8} M^i \psi(a_i)$ and $\alpha' =
\sum_{i=0}^{8} M^i \psi(a_i')$. Note that $\alpha$ and $\alpha'$ are in
$\mathfrak{D}_{9}$, so this is bounded by the maximum over all $u, v \in
\mathfrak{D}_{9}$ of $\abs{ \tau_{3}(u) - \tau_{3}(v) }$. It turns out
that there are only 301 vectors in $\mathfrak{D}_{9}$, so it is not
difficult for a computer program to determine that the maximum over all $u, v \in
\mathfrak{D}_9$ is achieved by
$u = (24, 30, 24, 12)$ and $v = (17, 25, 13, 5)$, and $\abs{ \tau_3(u)
- \tau_3(v) } \doteq  1.05517$.

For the infinite series, we first compute an upper bound for
$\abs{\tau_3(\delta_i)}$. Since $\delta_i = \psi(a_i') - \psi(a_i)$
where $a_i, a_i' \in \{ \varepsilon, 0, 4 \}$, it follows that
$\abs{\tau_3(\delta_i)}$ is less than the maximum over all $s, t \in \{
\varepsilon, 0, 4 \}$ of $\abs{\tau_3(\psi(s) - \psi(t))}$. The maximum
turns out to be $C = \abs{\tau_3(1,0,0,0)} \doteq 0.81582$, and is achieved
by $s = 0, t = \varepsilon$. Then
\begin{align*}
\abs{\sum_{i=9}^{\infty} \lambda_3^i \tau_3(\delta_i)} &\leq
\sum_{i=9}^{\infty} \abs{\lambda_3}^i \abs{\tau_3(\delta_i)} \leq C
\sum_{i=9}^{\infty} \abs{\lambda_3}^i = \frac{C \abs{\lambda_3}^9}{1 -
\abs{\lambda_3}} \doteq 0.032736 .
\end{align*}

Combining these two bounds gives us 
\begin{align*}
\abs{\tau_3(\psi(b))} &\leq \abs{\sum_{i=0}^{8} \lambda_3^i \delta_i} + \abs{\sum_{i=9}^{\infty} \lambda_3^i \delta_i} \leq 1.05517\cdots + 0.032736\cdots = 1.0879 \cdots = C_3/2.
\end{align*}

By the same reasoning, we get an identical bound $\abs{\tau_{3}(\psi(c))} \leq C_3/2$, and hence
\begin{align*}
\abs{\tau_3(\psi(b) - \psi(c))} &\leq \abs{\tau_3(\psi(b))} + \abs{\tau_3(\psi(c))} \leq C_3 \doteq 2.1758.
\end{align*}
\end{proof}

\subsection{Intersection with the lattice}
\label{subsec:intersec}

In this section, we start to use the fact that
the blocks $b_0$ and $c_0$ have the same length and sum. This is true if and only if 
\begin{align*} 
(1, 1, 1, 1) \cdot (\psi(b_0) - \psi(c_0)) &= 0 \\
(0, 1, 3, 4) \cdot (\psi(b_0) - \psi(c_0)) &= 0.
\end{align*}
Define a lattice of integer points that meet these conditions, as follows:
\begin{align*}
\mathfrak{L} &:= \{ \mathbf{v} \in \Zee^4 \ : \ (1,1,1,1) \cdot \mathbf{v} = 0
\text{ and } (0,1,3,4) \cdot \mathbf{v} = 0 \}.
\end{align*}
By construction, $\psi(b_0) - \psi(c_0)$ belongs to $\mathfrak{L}$.
If $\mathbf{v} = (v_1, v_2, v_3, v_4) \in \mathfrak{L}$,
then $(0,1,3,4) \cdot \mathbf(v) = 0$ implies $v_2 = -3v_3 - 4v_4$.
When we combine this with the equation $(1,1,1,1) \cdot \mathbf{v} = 0$
we get $v_1 = 2v_3 + 3v_4$. It follows that any vector in
$\mathfrak{L}$ is of the form
\begin{align*}
(2,-3,1,0)v_3 + (3,-4,0,1)v_4
\end{align*}
where $v_3, v_4 \in \Zee$. Since the vectors $(1,-2,2,-1)$ and
$(1,-1,-1,1)$ are an orthogonal basis for this lattice, every vector
in $\mathfrak{L}$ can be written as an integer linear combination of
these vectors.

Since $\psi(b_0) - \psi(c_0)$ is in $\mathfrak{L}$ we can write it in
this form. From Theorem~\ref{thm34} we know that $\abs{\tau_3(\psi(b_0) - \psi(c_0))}$ is
bounded, so we might expect $\psi(b_0) - \psi(c_0)$ to be a short
vector in the lattice. But this requires proof.

\begin{proposition}
Let $\mathbf{v} = m(1,-2,2,-1) + n(1,-1,-1,1)$ be an arbitrary vector in $\mathfrak{L}$. Then
\begin{align}
\abs{\tau_{3}(\mathbf{v})} &\geq \alpha \abs{m} \label{eq22}\\
\abs{\tau_{3}(\mathbf{v})} &\geq \beta \abs{n} .  \label{eq33}
\end{align}
where $\alpha \doteq 1.4914$ and $\beta \doteq 2.1657$ are constants. 
\label{prop:bound}
\end{proposition}
\begin{proof}
We prove only Eq.~(\ref{eq22}), leaving Eq.~(\ref{eq33}) to the reader.
The equation holds trivially when $m = 0$, so assume $m \neq 0$. We also use the fact that $\tau_{3}(1, -1, -1, 1) \doteq 0.80357 - 2.09082i$ is nonzero. Then
\begin{align*}
\abs{\tau_3 ( \mathbf{v} )} &= \abs{m \tau_3 (1, -2, 2, -1) + n \tau_3 (1, -1, -1, 1)} \\
&= \abs{m} \abs{\tau_3 (1,-1,-1,1)} \abs{\frac{n}{m} + \frac{\tau_3 (1,-2,2,-1)}{\tau_3 (1,-1,-1,1)}} \\
&\geq \abs{m} \abs{\tau_3 (1,-1,-1,1)} \abs{ \text{Im}\left(\frac{n}{m} + \frac{\tau_3 (1,-2,2,-1)}{\tau_3 (1,-1,-1,1)} \right)} \\
&\geq \abs{m} \abs{\tau_3 (1,-1,-1,1)} \abs{\text{Im}\left(\frac{\tau_3 (1,-2,2,-1)}{\tau_3 (1,-1,-1,1)} \right)} .
\end{align*}
If we take 
\begin{align*}
\alpha &= \abs{\tau_3 (1,-1,-1,1)} \abs{\text{Im}\left(\frac{\tau_3 (1,-2,2,-1)}{\tau_3 (1,-1,-1,1)} \right)} \doteq 1.4914,
\end{align*}
then $\abs{\tau_3(\mathbf{v})} \geq \alpha \abs{m}$, completing the proof. 
\end{proof}

Recall from last section that $\abs{\tau_3 (\psi(b_0) - \psi(c_0))} \leq 2.1758$. 
By Proposition~\ref{prop:bound}, if $\tau_3 (\psi(b_0) - \psi(c_0)) = m(1,-2,2,-1) + n(1,-1,-1,1)$ then
\begin{displaymath}
1.4914 \abs{m} \leq \abs{\tau_3 (\psi(b_0) - \psi(c_0))} \leq 2.1758,
\end{displaymath} 
so
$$|m| \leq \frac{2.1758}{1.4914} \dot= 1.4589.$$
Therefore $m \in \{ -1, 0, 1 \}$. Similarly, we deduce that $n \in \{ -1, 0, 1 \}$. 

Table~\ref{tab:lattice} lists all 9 possible vectors with $m, n \in
\{-1, 0, 1\}$. We see that only three of them, $(0,0,0,0)$,
$(1,-2,2,-1)$ and $(-1,2,-2,1)$, satisfy the constraint
$\abs{\tau_3(\psi(b_0) - \psi(c_0))} \leq 2.1758$, so $\psi(b_0) -
\psi(c_0)$ must be one of these vectors.

\begin{table}[H]
\begin{center}
\begin{tabular}{|c|c|c|c|c|c|} 
\hline
$v$ & $m$ & $n$ & $\abs{\tau_1(v)}$ & $\abs{\tau_2(v)}$ & $\abs{\tau_3(v)}$ \\ 
\hline $(0,0,0,0)$ & 0 & 0  & 0 & 0 & 0 \\
$(1,-2,2,-1)$ & $1$ & $0$   & 0.63278 & 1.51365 & 1.5425 \\ 
$(-1,2,-2,1)$ & $-1$ & $0$  & 0.63278 & 1.51365 & 1.5425 \\ 
$(1,-1,-1,1)$ & $0$ & $1$   & 0.21770 & 0.62031 & 2.23992 \\ 
$(-1,1,1,-1)$ & $0$ & $-1$  & 0.21770 & 0.62031 & 2.23992 \\ 
$(2,-3,1,0)$ & $1$ & $1$    & 0.41508 & 2.13396 & 2.37327 \\ 
$(-2,3,-1,0)$ & $-1$ & $-1$ & 0.41508 & 2.13396 & 2.37327 \\
$(0,1,-3,2)$ & $-1$ & $1$   & 0.85048 & 0.89334 & 3.02667 \\
$(0,-1,3,-2)$ & $1$ & $-1$  & 0.85048 & 0.89334 & 3.02667 \\ 
\hline
\end{tabular} \end{center} 
\caption{Short lattice vectors ordered by $\abs{\tau_{3}(v)}$.}
\label{tab:lattice}
\end{table}

\subsection{Bounding two more coordinates}

\begin{theorem}
\label{thm12}
Let $b$ and $c$ be blocks (not necessarily consecutive)
with the same length and same sum,
and let $\{ b_i \}_{i=0}^{\infty}$ and
$\{ c_i \}_{i=0}^{\infty}$ be their ancestral sequences respectively. Then
\begin{align*}
\abs{\tau_1(\psi(b_i) - \psi(c_i))} &\leq C_1 \\
\abs{\tau_2(\psi(b_i) - \psi(c_i))} &\leq C_2
\end{align*}
for all $i$, where
\begin{align*}
C_1 &= \frac{2 \abs{\tau_1(0,0,0,1)}}{\abs{\lambda_1} - 1} \doteq 1.9032 \\
C_2 &= \frac{2 \abs{\tau_2(1,0,0,-1)}}{\abs{\lambda_2} - 1} \doteq 2.9818.
\end{align*}
\end{theorem} 
\begin{proof}
Our proof is inductive, starting at $i = 0$. In the last section, we argued that $\psi(b_0) - \psi(c_0)$ is either $(0,0,0,0)$, $(1,-2,2,-1)$ or $(-1,2,-2,1)$. We see from table~\ref{tab:lattice} that both inequalities are satisfied for these vectors. 

Otherwise, 
\begin{align*}
\abs{\tau_2(\psi(b_{i+1}) - \psi(c_{i+1}))} &= \abs{\tau_2\left(  M^{-1}(\psi(b_{i}) - \psi(c_i) - \delta_{i} + \delta_{i}') \right)} \\
&= \abs{\lambda_2}^{-1} \abs{\tau_2(\psi(b_{i}) - \psi(c_i) - \delta_{i} + \delta_{i}')} \\
&\leq \abs{\lambda_2}^{-1} \left(  \abs{\tau_2(\psi(b_{i}) - \psi(c_i))} + \abs{\tau_2(\delta_{i})} + \abs{\tau_2(\delta_{i}')} \right) .
\end{align*}
As in the proof of Theorem~\ref{thm34}, the quantities
$\abs{\tau_2(\delta_i)}$ and $\abs{\tau_2(\delta_i')}$ are bounded by
the maximum over all $u, v \in \{ \varepsilon, 0, 4 \}$ of
$\abs{\tau_2(u-v)}$, which turns out to be $\abs{\tau_2(1,0,0,-1)} \doteq
0.75301$. By induction, we know that $\abs{\tau_2(\psi(b_{i}) -
\psi(c_i))} \leq \frac{2\abs{\tau_2(1,0,0,-1)}}{\abs{\lambda_2} - 1}$.
Therefore we have
\begin{align*}
\abs{\tau_2(\psi(b_{i+1}) - \psi(c_{i+1}))} &\leq \abs{\lambda_2}^{-1}
\left(  \frac{2\abs{\tau_2(1,0,0,-1)}}{\abs{\lambda_2} - 1} +
2\abs{\tau_2(1,0,0,-1)} \right) =
\frac{2\abs{\tau_2(1,0,0,-1)}}{\abs{\lambda_2} - 1} .
\end{align*}
This completes the induction, and the proof of the second inequality. The proof of the first inequality is virtually identical, so it is left to the reader.
\end{proof}

We have now bounded all four eigencoordinates of $\psi(b_{i}) - \psi(c_i)$. 
In other words, $\tau( \psi(b_{i}) - \psi(c_i) )$ has bounded length.
Since $\tau$ is invertible,
$\psi(b_{i}) - \psi(c_i)$ will also have bounded length. 

\subsection{Finding $\mathfrak{U}$}

Define the set
\begin{align*}
\mathfrak{U} &= \{ \mathbf{x} \in \Zee^4 \ : \  \abs{\tau_i(\mathbf{x})} \leq C_i 
\text{ for $1 \leq i \leq 4$ } \}
\end{align*}
where $C_1, C_2, C_3, C_4$ are the constants in Theorems~\ref{thm34} and \ref{thm12}. These 
theorems prove that $\psi(b_i) - \psi(c_i)$ is an element of $\mathfrak{U}$. 

\begin{proposition}
If $\mathbf{x} \in \mathfrak{U}$ then $\abs{\mathbf{x}} \leq 6.28$. 
\label{prop:ten}
\end{proposition}

\begin{proof}
Since $\mathbf{x}$ is in $\mathfrak{U}$, there are bounds on the components of
$\tau(\mathbf{x})$, and therefore on its length. We have
\begin{align*}
\abs{\tau(\mathbf{x})}^2 &= \sum_{i=1}^{4} \abs{\tau_i(\mathbf{x})}^2 \\
&\leq 1.9032^2 + 2.9818^2 + 2.1758^2 + 2.1758^2 \\
&\doteq 21.98 .
\end{align*}
Suppose $A$ is an arbitrary matrix and every eigenvalue $\lambda$ of $A^{*} A$
lies between $\mu_{\min}$ and $\mu_{\max}$, where $A^{*}$ denotes the conjugate
transpose of $A$. Under these conditions, a classical theorem in 
linear algebra (e.g., \cite[pp.\ 415--420]{Atkinson:1978})
gives the inequality
\begin{align*}
\mu_{\min} \abs{x}^2 &\leq \abs{Ax}^2 \leq \mu_{\max} \abs{x}^2  .
\end{align*}
The smallest eigenvalue of $\tau^{*} \tau$ is $\mu \doteq
0.55713$, so from the lower bound of the theorem we get
$$\abs{\mu} \abs{\mathbf{x}}^2 \leq \abs{\tau(\mathbf{x})}^2 $$
and so
$$ \abs{\mathbf{x}}^2 \leq \frac{\abs{\tau(\mathbf{x})}^2}{\abs{\mu}} \doteq \frac{21.98}{0.55713} \doteq 39.455 .$$
Therefore $\abs{\mathbf{x}} \leq 6.28$.
\end{proof}

Proposition~\ref{prop:ten} tells us that we can enumerate vectors in
$\mathfrak{U}$ by listing all integer vectors of length less than $6.28$ and
discarding the ones that fail to satisfy the inequalities
$\abs{\tau_i(\mathbf{x})} \leq C_i$ for each $i$. Our computer program for enumerating 
$\mathfrak{U}$ lists 503 vectors.  There is not enough space to
reproduce the entire list here, but it can be downloaded from
{\tt http://www.student.cs.uwaterloo.ca/\char'176l3schaef/sumcube/ }.

\section{Main graph} \label{sec:graph}

\subsection{Graph products}

Recall the tree $\mathcal{T}$ with vertices representing positions in
$\mathbf{w}$. We are interested in finding an analogous graph for triple
blocks. Since each triple block is delimited by four positions, so it is
natural to consider the graph $\mathcal{T} \times \mathcal{T} \times
\mathcal{T} \times \mathcal{T} = \mathcal{T}^4$ where the product of two graphs
is defined below. 
\begin{definition}
Let $\mathcal{G}_1$ and $\mathcal{G}_2$ be graphs. Define the tensor product
$\mathcal{G}_1 \times \mathcal{G}_2$ where
\begin{align*}
V(\mathcal{G}_1 \times \mathcal{G}_2) &= V(\mathcal{G}_1) \times V(\mathcal{G}_2) \\
E(\mathcal{G}_1 \times \mathcal{G}_2) &= E(\mathcal{G}_1) \times E(\mathcal{G}_2)
\end{align*}
Further, if $e_1$ and $e_2$ are labelled by $\ell_1$ and $\ell_2$ respectively
then the edge $(e_1, e_2)$ is labelled $(\ell_1, \ell_2)$.
\end{definition}
Note that every vertex $\mathbf{p} = (p_1, p_2, p_3, p_4)$ in $\mathcal{T}^4$
has a unique parent $\parent{\mathbf{p}}$ because each coordinate $p_i$ has a
unique parent $\parent{p_i}$ in $\mathcal{T}$. Thus, if a vertex $v$ in
$\mathcal{T}^4$ delimits a triple block $b_1 b_2 b_3$ then the parent of $v$
delimits the parent of $b_1 b_2 b_3$. 

Since each node in $\mathcal{T}^4$ has a unique parent, we define the ancestral
sequence of a node $v$ to be the sequence of parents starting at the node. An
ancestral sequence in $\mathcal{T}^4$ gives us four ancestral sequences in
$\mathcal{T}$ for the four coordinates of $\mathcal{T}^4$. Every ancestral
sequence in $\mathcal{T}$ eventually reaches 0, the root, and remains there
because $0$ is its own parent. It follows that every ancestral sequence in
$\mathcal{T}^4$ eventually reaches $(0,0,0,0)$. Also note that any cycle in
$\mathcal{T}^4$ induces four cycles in $\mathcal{T}$. Since the loop from $0$
to itself is the only cycle in $\mathcal{T}$, it follows that the only cycle in
$\mathcal{T}^4$ is a loop from $(0,0,0,0)$ to itself. We conclude that
$\mathcal{T}^4$ is a tree, except for a loop at $(0,0,0,0)$, the root of the
tree. 

Earlier we saw that the graph homomorphism $\zeta : \mathcal{T} \rightarrow
\mathcal{Q}$ was child-bijective. By Proposition~\ref{bijectwalk} there is a
certain relationship between walks in $\mathcal{T}$ and walks in $\mathcal{Q}$.
Define a graph homomorphism $\xi : \mathcal{T}^4 \rightarrow \mathcal{Q}^4$
that sends $(p_1, p_2, p_3, p_4)$ to $(\zeta(p_1), \zeta(p_2), \zeta(p_3),
\zeta(p_4))$. By the following proposition, $\xi$ is a child-bijection. 
\begin{proposition}
If $f_1 \colon \mathcal{A}_1 \rightarrow \mathcal{B}_1$ and $f_2 \colon
\mathcal{A}_2 \rightarrow \mathcal{B}_2$ are child-bijections then $f \colon
\mathcal{A}_1 \times \mathcal{A}_2 \rightarrow \mathcal{B}_1 \times
\mathcal{B}_2$ sending $(x, y)$ to $(f_1(x), f_2(y))$ is also a
child-bijection. 
\end{proposition}
\begin{proof}
Exercise.
\end{proof}
Since $\xi$ is a child-bijection, Proposition~\ref{bijectwalk} says that if we
fix some $v$ in $\mathcal{T}^4$ then there is a bijection between walks in
$\mathcal{T}^4$ starting at $v$ and walks in $\mathcal{Q}^4$ starting at
$\xi(v)$. Furthermore, the bijection sends a walk $v_0 \cdots v_\ell$ to a walk
$\xi(v_0) \cdots \xi(v_\ell)$, and the bijection preserves edge labels.

\subsection{Augmenting $\mathcal{Q}$ with block vectors}

We know that a triple block $b_1 b_2 b_3$ delimited by $(p_1, p_2, p_3, p_4)$
will give us a walk from $(0,0,0,0)$ to $(p_1, p_2, p_3, p_4)$ in
$\mathcal{T}^4$. Then we know the Parikh vector of each block, 
\begin{align*}
\psi(b_1) &= \sigma(p_2) - \sigma(p_1) \\
\psi(b_2) &= \sigma(p_3) - \sigma(p_2) \\
\psi(b_3) &= \sigma(p_4) - \sigma(p_3).
\end{align*}
If $b_1 b_2 b_3$ is an additive cube, then
\begin{align*}
\psi(b_2) - \psi(b_1) &= \sigma(p_3) - 2\sigma(p_2) + \sigma(p_1) \\
\psi(b_3) - \psi(b_2) &= \sigma(p_4) - 2\sigma(p_3) + \sigma(p_2)
\end{align*}
must both belong to $\mathfrak{L}$ and, since the blocks are nonempty, we have
$p_1 < p_2 < p_3 < p_4$. The point is that we can decide whether a vertex in
$\mathcal{T}^4$ corresponds to an additive cube. 

The problem with $\xi$ is that we cannot tell from $\xi(p_1, p_2, p_3, p_4)$
whether $(p_1, p_2, p_3, p_4)$ delimits an additive cube, because the vertices in
$\mathcal{Q}^4$ do not contain any information about the Parikh vectors. On the
other hand, suppose we project the walk from $(0,0,0,0)$ to $(p_1, p_2, p_3,
p_4)$ into $\mathcal{Q}^4$ via $\xi$. Then since mapping walks via $\xi$ is
bijective, we can recover original walk in $\mathcal{T}^4$ from its image in
$\mathcal{Q}^4$. 

Now let there be a walk from $(p_1, \ldots, p_4)$ to $(q_1, \ldots, q_4)$ in
$\mathcal{T}^4$, and suppose we are given the corresponding walk in
$\mathcal{Q}^4$. We have seen that if we have $(p_1, \ldots p_4)$ then we can
reconstruct the entire original walk $\mathcal{T}^4$, including $(q_1, \ldots,
q_4)$. What do we need to know about $(p_1, \ldots, p_4)$ to reconstruct just
$\sigma(q_3) - 2\sigma(q_2) + \sigma(q_1)$? Let us suppose for simplicity that
the walk is length 1, so $(p_1, \ldots, p_4)$ and $(q_1, \ldots, q_4)$ are
adjacent, and let the edge be labelled $(a_1, \ldots, a_4)$. Then by
Lemma~\ref{sigmarec},
\begin{align*}
\sigma(q_3) - 2\sigma(q_2) + \sigma(q_1) &= (M \sigma(p_3) + \psi(a_3)) - 2(M \sigma(p_2) + \psi(a_2)) + (M \sigma(p_1) + \psi(a_1)) \\
&= M (\sigma(p_3) - 2 \sigma(p_2) + \sigma(p_1)) + \psi(a_3) - 2 \psi(a_2) + \psi(a_1),
\end{align*}
so it suffices to know $\sigma(p_3) - 2 \sigma(p_2) + \sigma(p_1)$. Similarly,
we can compute $\sigma(q_4) - 2\sigma(q_3) + \sigma(q_2)$ from $\sigma(p_4) -
2\sigma(p_3) + \sigma(p_2)$. 

Thus, we define the graph $\mathcal{G}$ with vertices \footnote{The vertices of $\mathcal{G}$ resemble \emph{templates} as defined in \cite{Aberkane&Currie&Rampersad:2004}. Templates and $\mathcal{G}$ are different perspectives on the same structure.} in the set
\begin{align*}
V(\mathcal{G}) &= V(\mathcal{Q}^4) \times \Zee^4 \times \Zee^4 = \Sigma^4 \times \Zee^4 \times \Zee^4. 
\end{align*}
There is an edge from $(\mathbf{c}, \mathbf{u}, \mathbf{v})$ to $(\mathbf{c}',
\mathbf{u}', \mathbf{v}')$ if there is an edge $(\mathbf{c}, \mathbf{c}')$ in
$\mathcal{Q}^4$ labelled by $(a_1, a_2, a_3, a_4)$ and the following
consistency equations hold:
\begin{align*}
\mathbf{u}' &= M \mathbf{u} + \psi(a_3) - 2 \psi(a_2) + \psi(a_1) \\
\mathbf{v}' &= M \mathbf{v} + \psi(a_4) - 2 \psi(a_3) + \psi(a_2) .
\end{align*}
To see where these equations come from, define the map $g \colon \mathcal{T}^4
\rightarrow \mathcal{G}$ such that 
\begin{align*}
g(\mathbf{p}) &= (\xi(\mathbf{p}), \sigma(p_3) - 2 \sigma(p_2) + \sigma(p_1), \sigma(p_4) - 2 \sigma(p_3) + \sigma(p_2)).
\end{align*}
Then a triple block $b_1 b_2 b_3$ maps to the four delimiting characters
$w[p_1], w[p_2], w[p_3], w[p_4]$, with the vectors $\psi(b_2) - \psi(b_1)$ and
$\psi(b_3) - \psi(b_2)$. Thus, we see that the two vectors are meant to
represent the differences between pairs of block vectors. We have seen that
given an edge $(\xi(\mathbf{p}), \xi(\mathbf{q}))$ in $\mathcal{Q}^4$, we may
compute $\sigma(q_{i+2}) - 2 \sigma(q_{i+1}) + \sigma(q_{i})$ from
$\sigma(p_{i+2}) - 2 \sigma(p_{i+1}) + \sigma(p_{i})$ for $i = 1, 2$. The
consistency equations simply enforce this relationship, so that
$(g(\mathbf{p}), g(\mathbf{q}))$ is an edge in $\mathcal{G}$ when $(\mathbf{p},
\mathbf{q})$ is an edge in $\mathcal{T}^4$. Thus, $g$ is a graph homomorphism
from $\mathcal{T}^4$ to $\mathcal{G}$. 

\begin{proposition}
The graph homomorphism $g$ is child-bijective.
\end{proposition}
\begin{proof}
Consider an edge $(g(\mathbf{p}), (\mathbf{c}, \mathbf{u}, \mathbf{v}))$ in
$\mathcal{G}$. By the definition of edges in $\mathcal{G}$ there must be an
edge $(\xi(\mathbf{p}), \mathbf{c})$ in $\mathcal{Q}^4$. Since $\xi$ is
child-bijective, there exists a neighbour $\mathbf{q}$ of $\mathbf{p}$ in
$\mathcal{T}^4$ such that $\xi(\mathbf{q}) = \mathbf{c}$. Then $(g(\mathbf{p}),
g(\mathbf{q}))$ is an edge in $\mathcal{G}$. The consistency equations have a
unique solution for $g(\mathbf{q})$ given $g(\mathbf{p})$ and $\mathbf{c}$, so
$(\mathbf{c}, \mathbf{u}, \mathbf{v}) = g(\mathbf{q})$ as required. 
\end{proof}

Since $g$ is child-bijective, Proposition~\ref{bijectwalk} says that any walk
in $\mathcal{G}$ starting at $g(\mathbf{p})$ is the image of some walk from
$\mathbf{p}$ to $\mathbf{q}$ in $\mathcal{T}^4$, and therefore ends in
$g(\mathbf{q})$.

\subsection{Additive cubes and walks in $\mathcal{G}$}

We saw in the previous section that walks in $\mathcal{T}^4$ correspond to
walks in $\mathcal{G}$. If an arbitrary triple block is delimited by
$\mathbf{q}$ then there is a walk in $\mathcal{T}^4$ from $(0,0,0,0)$ to
$\mathbf{q}$ and therefore a walk from $g(0,0,0,0)$ to $g(\mathbf{q})$. 
Unfortunately, a walk from $g(0,0,0,0)$ to $g(\mathbf{q})$ in $\mathcal{G}$
does not necessarily mean that $\mathbf{q}$ delimits a triple block, since we
do not necessarily have $q_1 < q_2 < q_3 < q_4$. We can fix this, but we need a
few propositions.

Note that $\parent{x}$ is an increasing function, so if $p \leq q$ then
$\parent{p} \leq \parent{q}$. By induction, if $p_0 \leq q_0$ then for $i$th
ancestors $p_i$ and $q_i$ we have $p_i \leq q_i$. The contrapositive says that
$p_i < q_i$ implies $p_0 < q_0$, which has an application in the following
lemma. 
\begin{lemma}
Given an ancestral sequence of blocks $\{ b_i \}_{i=0}^{\infty}$ for some
nonempty block $b_0$, there exists $k$ such that $b_j$ is nonempty for all $0
\leq j \leq k$ and $b_j$ is empty for all $k < j$. Furthermore, $b_k$ is the
difference of two proper prefixes of some $\varphi(c)$, so $b_k$ is either $0$
or $4$. 
\end{lemma}
\begin{proof}
Suppose $b_0 = w[p_0, q_0)$ and let $\{ p_i \}_{i=0}^{\infty}$ and $\{ q_i
\}_{i=0}^{\infty}$ be the corresponding ancestral sequences for $p_0$ and
$q_0$. Recall that the ancestral sequence of a block $b_0$ is the sequence of
parents of $b_0$, so
\begin{align*}
b_i = \parent{b_{i-1}} = {\bf w}[\parent{p_{i-1}}, \parent{q_{i-1}}) = {\bf w}[p_i, q_i).
\end{align*}
An ancestral sequence of positions eventually reaches 0 so there exists some
$n_0 \in \Enn$ such that $p_n = q_n = 0$, and hence $b_n$ is empty, for all $n
> n_0$. Thus, we can take $b_k$ to be the last nonempty block because there are
only finitely many nonempty blocks. By definition, $b_j$ is empty for all $j >
k$. Since $b_k$ is nonempty we have $p_k < q_k$, and thus $p_j < q_k$ for any
$j \leq k$ by the discussion above. It follows that $b_j$ is nonempty for any
$j \leq k$. 

By the definition of parents, $\eta(q_{k+1}) \leq q_k < \eta(q_k)$. It follows
that $y = {\bf w}[\eta(q_{k+1}), q_k)$ is a proper prefix of ${\bf
w}[\eta(q_{k+1}), \eta(q_{k+1} + 1)) = \varphi(w[q_{k+1}])$. Similarly, $x =
{\bf w}[\eta(p_{k+1}), p_k)$ is a proper prefix of ${\bf w}[\eta(p_{k+1}),
\eta(p_{k+1} + 1)) = \varphi(w[p_{k+1}])$. But $p_{k+1} = q_{k+1}$ because
$b_{k+1}$ is empty, so $\eta(p_{k+1}) = \eta(q_{k+1})$ and the prefixes $x$ and
$y$ start at the same position and $x b_k = y$. For our morphism, we have seen
that the proper prefixes are $\{ \varepsilon, 0, 4 \}$. Then $b_k$ is a suffix
of some proper prefix, and since $b_k$ is nonempty we have $b_k \in \{ 0, 4
\}$. 
\end{proof}

Suppose we have a triple block delimited by $\mathbf{p}$. Usually we
are interested in the three blocks inside, but we can also think of it
as one big block delimited by $p_1$ and $p_4$. Then the lemma applies,
and motivates the following definition:
\begin{align*}
X := \{ (p_1, p_2, p_3, p_4) \in \Enn^4 \, : \, p_1 \leq p_2 \leq p_3 \leq p_4, p_1 < p_4, \parent{p_1} = \parent{p_2} = \parent{p_3} = \parent{p_4} \}.
\end{align*}
We have defined $X$ so that it is precisely the set of positions that delimit a nonempty triple block with an empty parent.

Now let us consider the set $g(X)$.  Since $\mathbf{p}$ is in $X$, we know
from the lemma 
that it delimits a triple block $b_1 b_2 b_3$ that is a subword of
$\varphi(c)$ for some $c$. Since $\varphi(0), \varphi(1), \varphi(3),
\varphi(4)$ all occur in the first seven characters of $\mathbf{w}$, and $g$
depends on the content of the triple block and not its absolute position, it
suffices to compute $g(\mathbf{p})$ for all $\mathbf{p} \in X$ such that $p_4 <
7$. The lemma states that $b_1 b_2 b_3$ is either $0$ or $4$, so the triple
block is one of $w[0, 1), w[3, 4), w[5, 6)$. Thus, 
\begin{align*}
g(X) = \{ & g(0,0,0,1), g(0,0,1,1), g(0,1,1,1), \\
& g(3,3,3,4), g(3,3,4,4), g(3,4,4,4), \\
& g(5,5,5,6), g(5,5,6,6), g(5,6,6,6) \}.
\end{align*}

A walk in $\mathcal{T}^4$ from $(0,0,0,0)$ to $\mathbf{q}$, where $\mathbf{q}$
delimits a nonempty 
triple block, must pass through the set $X$. Thus, we might consider
starting our walk from a node in $X$ instead of from $(0,0,0,0)$. If our walk
starts at some $\mathbf{p} \in X$ then we have $p_1 < p_4$ and hence $q_1 <
q_4$, which helps us to show that $q_1 < q_2 < q_3 < q_4$. Hence, additive
cubes
correspond to walks as described in Theorem~\ref{sumcubestowalks}. 
\begin{theorem}
\label{sumcubestowalks}
Given an additive
cube delimited by $\mathbf{q}$, there is a walk in $\mathcal{G}$
from $g(\mathbf{p})$ to $g(\mathbf{q})$, where $\mathbf{p}$ is in $X$ and
$g(\mathbf{q})$ is in the set $Z := \Sigma^4 \times \mathfrak{L} \times
\mathfrak{L}$. Conversely, if we are given a walk in $\mathcal{G}$ starting in
$g(X)$ and ending at $v \in Z$, there is an additive
cube delimited by $\mathbf{q}$
where $g(\mathbf{q}) = v$.
\end{theorem}
\begin{proof}
Suppose $\mathbf{q}$ delimits an additive
cube. Then the difference between two block
vectors is in $\mathfrak{L}$, so $g(\mathbf{q})$ is in $Z$. By the previous
lemma, there exists some ancestor $\mathbf{p} \in \mathcal{T}^4$ of
$\mathbf{q}$ such that $\mathbf{p}$ delimits a nonempty triple block, but
$\parent{\mathbf{p}}$ delimits $\varepsilon \varepsilon \varepsilon$. By our
construction of $X$, $\mathbf{p}$ belongs to $X$, and our walk from
$\mathbf{p}$ to $\mathbf{q}$ in $\mathcal{T}^4$ maps to a walk from
$g(\mathbf{p})$ to $g(\mathbf{q})$ in $\mathcal{G}$. 

In the other direction, suppose we have a walk in $\mathcal{G}$ starting at
$g(\mathbf{p})$ and ending at some vertex $v$ in $Z$. The walk in $\mathcal{G}$
corresponds to a walk in $\mathcal{T}^4$ from $\mathbf{p}$ to some $\mathbf{q}$
and $v = g(\mathbf{q})$. Now $g(\mathbf{q})$ is in $Z$, so the block vector
differences are in $\mathfrak{L}$, and therefore all blocks have the same
length and sum. Since $p_1 < p_4$ we know that $q_1 < q_4$ and thus every block
has positive length giving $q_1 < q_2 < q_3 < q_4$. We conclude that
$\mathbf{q}$ delimits an additive cube. 
\end{proof}

\subsection{Reduction to a finite subgraph}

The relationship between additive cubes and walks in $\mathcal{G}$ is nice, but
$\mathcal{G}$ is an infinite graph because $\Zee^4$ is infinite, and it is
difficult to apply graph algorithms to an infinite graph. Recall the main
result of the previous section, which states that if $b_0$ and $c_0$ are blocks
with the same length and sum, and $\{ b_i \}_{i=0}^{\infty}, \{ c_i
\}_{i=0}^{\infty}$ are their ancestral sequences respectively then $\psi(b_i) -
\psi(c_i) \in \mathfrak{U}$, where $\mathfrak{U}$ is a finite set we can
enumerate. In other words, if $\mathbf{p}$ delimits an
additive cube and $\mathbf{q}$
is any node on the walk from $(0,0,0,0)$ to $\mathbf{p}$, and $\mathbf{q}$
delimits blocks $d_1 d_2 d_3$ then $\psi(d_i) - \psi(d_j) \in \mathfrak{U}$ for
all $i, j$. Let $g(\mathbf{q}) = (\mathbf{c}, \mathbf{u}, \mathbf{v})$ and note
that by the definition of $g$, we have $\mathbf{u} = \psi(d_2) - \psi(d_1)$,
$\mathbf{v} = \psi(d_3) - \psi(d_2)$ and finally $\mathbf{u} + \mathbf{v} =
\psi(d_3) - \psi(d_1)$. It follows that $\mathbf{u}$, $\mathbf{v}$, and
$\mathbf{u}+\mathbf{v}$ are all in $\mathfrak{U}$. 

Define the set 
$$
H := \{ (\mathbf{c}, \mathbf{u}, \mathbf{v}) \ :\  \mathbf{u}, \mathbf{v}, \mathbf{u}+\mathbf{v} \in \mathfrak{U} \text{ and } \mathbf{c} \in \Sigma^4 \}.
$$
Any node along the path to an additive cube must be in $H$, so we may restrict our
search to the subgraph of $\mathcal{G}$ induced by $H$, call it $\mathcal{G}'$.
Notice that $H$ is a subset of $\Sigma^4 \times \mathfrak{U} \times
\mathfrak{U}$, so it contains at most $4^4 \times 503 \times 503 \doteq 64.7
\times 10^6$ elements and thus $\mathcal{G}'$ is finite. We can update
Theorem~\ref{sumcubestowalks} so that the walks must lie in $\mathcal{G}'$.
\begin{corollary}
Given an additive cube delimited by $\mathbf{q}$, there is a walk in $\mathcal{G}'$
from $g(\mathbf{p})$ to $g(\mathbf{q})$, where $\mathbf{p}$ is in $A := g(X)
\cap H$ and $g(\mathbf{q})$ is in the set $B := Z \cap H$. Conversely, if we
are given a walk in $\mathcal{G}$ starting in $A$ and ending at some $\beta \in
B$, there is an additive cube delimited by $\mathbf{q}$ where $g(\mathbf{q}) =
\beta$.
\end{corollary}

Since $\mathcal{G}'$ is a practical size, we can determine whether there is a
path from $A$ to $B$ within $\mathcal{G}'$ using a computer, but first we need
to be able to compute $A$, $B$ and $H$. Since we can enumerate the set
$\mathfrak{U}$, it is straightforward to enumerate $\Sigma^4 \times
\mathfrak{U} \times \mathfrak{U}$ and then narrow it down to $H$. Also, we can
test whether a vector is in the lattice $\mathfrak{L}$ using dot products and
therefore we can determine if an element is in $Z = \Sigma^4 \times
\mathfrak{L} \times \mathfrak{L}$. Thus, we can list the elements of $B = Z
\cap H$ by testing whether each element in $H$ is also in $Z$. Earlier we
showed that $g(X)$ is a nine-element set, and it is not difficult to check that
$g(X) \subseteq H$, so $A = g(X) \cap H = g(X)$. 

We are now ready to complete the proof of our main result.
\begin{theorem}
The infinite word $\bf w$ contains no additive cubes.
\end{theorem}
\begin{proof}
Our computer search of $\mathcal{G}'$ computes the set of vertices $R \subseteq
H$ that are reachable from $A$. It turns out that $\abs{R} = 135572$, but $R
\cap B$ is empty, so we conclude that $\bf w$ contains no additive cubes. 
\end{proof}

\begin{corollary}
$\Enn^+$ is not uniformly $3$-repetitive.
\end{corollary}

\section{A two-sided infinite word avoiding additive cubes}

Above we have shown that the (one-sided) infinite word 
$$ \phiright(0) = 031430110343430 \cdots$$
avoids additive cubes.   Using the appropriate iterations of $\varphi$,
we now prove the same result for two-sided
infinite words.  (Such a word is a map from $\Zee$ to a finite
set, in this case, $\Sigma = \lbrace 0, 1, 3, 4 \rbrace$, as opposed
to the one-sided infinite words --- maps from $\Enn$ to $\Sigma$ --- we
have discussed thus far.)  

For notation and results involving morphisms and two-sided infinite
words, see \cite{Shallit&Wang:2002}.  In particular, we write
a two-sided infinite word as $\cdots a_{-2} a_{-1} a_0 . a_1 a_2 \cdots$,
where the
period is written to the left of the character indexed with $1$.
Also, if $h$ is a morphism satisfying $h(a) = xa$ for some word $x$,
then we define $\hright(a)$ to be the left-infinite word
$\cdots h^3(x) h^2(x) h(x) x a $.

\begin{theorem}
There exists a two-sided infinite word over
$\Sigma= \lbrace 0, 1, 3, 4 \rbrace$ avoiding additive cubes.
\end{theorem}

\begin{proof}
Note that $30$ is a factor of $\varphi^4 (0)$.  It follows
that for all $n \geq 0$, the word $\varphi^n (30)$ contains no
additive cube.  Now $\varphi^2(3) = 43$, and
$\varphi^2(0) = 031$.
Letting $h = \varphi^2$, we see that
$$ \hleft(3) \ . \hright(0) =  \cdots 03143034343034343.03143011034343031011011 \cdots $$
is a two-sided infinite word avoiding additive cubes.
\end{proof}

\section{Open problems}

We do not know if the alphabet size of $4$ given in this paper is optimal
for avoiding additive cubes.  Since an abelian cube is necessarily an additive cube,
and we know it is impossible to avoid abelian cubes over an alphabet of
size $2$, the alphabet size cannot be $2$.  But it is still conceivable
that, to avoid additive cubes,  some alphabet of cardinality $3$ might suffice.
By a depth-first search approach, for example, we have generated a finite
word of length 1288 over the alphabet $\lbrace 0,1,2 \rbrace$ avoiding
additive cubes.

The more difficult question of whether it
is possible to avoid additive squares over an alphabet equal to some
finite subset of $\Zee$ remains open.   Since alphabet size $4$ is needed
to avoid abelian squares, the alphabet must be at least this large.  However,
as Freedman has shown \cite{Freedman:2010}, 
the longest word over $\lbrace a, b, c, d \rbrace$ with
$a+d = b+c$
avoiding additive squares is of length $\leq 60$.  Also see \cite{Justin:1972}.

\section{Acknowledgments}

We are pleased to thank Thomas Stoll for many discussions on this problem.

\newcommand{\noopsort}[1]{} \newcommand{\singleletter}[1]{#1}

\end{document}